\documentclass[twoside,12pt]{article}
\usepackage{graphicx,amsmath,amsfonts,amsthm,amssymb,fancyhdr}
\newtheorem{thm}{Theorem}[section]

\newtheorem{lem}[thm]{Lemma}
\newtheorem{prop}[thm]{Proposition}
	
\theoremstyle{definition}
	
\theoremstyle{remark}

\def\beq{\begin{eqnarray}}
\def\eeq{\end{eqnarray}}
\def\bsp{\begin{split}}	
\def\esp{\end{split}}

\newcommand{\be}{\begin{equation}}
\newcommand{\ee}{\end{equation}}

\newcommand{\al}{\alpha}
\newcommand{\bt}{\beta}

\def \hbm #1 {\mbox{\boldmath{$\hat m^{(#1)}$}}}

\def \bm #1 {\mbox{\boldmath{$m^{(#1)}$}}}
\def \BDM {\begin{displaymath}}
\def \EDM {\end{displaymath}}

\def\bfn{{\bf n}}

\begin{document}

\title{Isometries in Higher Dimensional $CCNV$ Spacetimes}
\author{{\large\textbf{D. McNutt$^{1}$,  A. Coley$^{1}$, N. Pelavas$^{1}$}}
\vspace{0.3cm} \\
$^{1}$Department of Mathematics and Statistics,\\
Dalhousie University,\\
Halifax, Nova Scotia,\\
Canada B3H 3J5
\vspace{0.5cm}\\
\texttt{ddmcnutt@dal.ca, aac@mathstat.dal.ca, pelavas@mathstat.dal.ca} }
\date{\today}
\maketitle
\pagestyle{fancy}
\fancyhead{}
\fancyhead[CE]{D. McNutt, A. Coley, N. Pelavas}
\fancyhead[LE,RO]{\thepage}
\fancyhead[CO]{Isometries in Higher Dimensional CCNV Spacetimes}
\fancyfoot{}

\begin{abstract}

We study the class of higher-dimensional Kundt metrics admitting a covariantly constant null vector, known as $CCNV$ spacetimes. We pay particular attention to those $CCNV$ spacetimes with constant (polynomial) curvature invariants ($CSI$). We investigate the existence of an additional isometry in $CCNV$ spacetimes, by studying the Killing equations for the general form of the $CCNV$ metric. In particular, we list all $CCNV$ spacetimes allowing an additional non-spacelike isometry for all values of the lightcone coordinate $v$, which are of interest due to the invariance of the metric under a translation in $v$. As an application we use our results to find all $CSI$ $CCNV$ spacetimes with an additional isometry as well as the subset of these spacetimes in which the isometry is non-spacelike for all values $v$.

\end{abstract}


\begin{section}{Introduction}

In this paper we will study the N-dimensional $CCNV$ spacetimes, which are defined as metrics admitting a covariantly constant null vector ($CCNV$). We say a spacetime is a constant scalar invariant spacetime ($CSI$) if all of the polynomial scalar invariants constructed from the Riemann tensor and its covariant derivatives are constant. Those spacetimes for which all of these polynomial scalar invariants vanish will be called  $VSI$. In particular, we discuss the subset of $CCNV$ spacetimes which are also $CSI$ or $VSI$.

Previously it was shown that the higher-dimensional  $VSI$  spacetimes with fluxes and dilaton are solutions of type IIB supergravity \cite{VSISUG}. Those $VSI$ spacetimes in which supersymmetry is preserved admit a $CCNV$ as such they are $CCNV$ spacetimes.  Similarly $CSI$ spacetimes have been shown to be solutions of supergravity \cite{CFH} as well, however there is less known about the conditions these $CSI$ spacetimes must satisfy in order to preserve some supersymmetry. Unlike $VSI$ spacetimes, the $CSI$ spacetimes do not have a canonical form for the metric tensor.

A necessary (but not sufficient) condition for supersymmetry to be preserved is that the spacetime admits a Killing spinor and hence a null or timelike Killing vector. In the case of $CCNV$ spacetimes we know their metrics must belong to the Kundt class because they contain a null vector which is geodesic, non-expanding, shear-free and non-twisting; namely, the covariantly constant null vector $\ell$. Using this fact, we will study the curvature of a general $CCNV$ spacetime, from which we will find conditions on the transverse space in order for a $CCNV$ spacetime to be $CSI$.

With the connection components determined, we examine the existence of an additional Killing vector in a general $CCNV$ spacetime. Supposing it admits another Killing vector we find constraints on the metric functions $H$, $W_e$, $m_{ie}$ and the Killing vector components. The result of this analysis determines the three general forms an additional Killing vector can take in a $CCNV$ spacetime, we consider the Killing Lie Algebra generated by $\ell$ and $X$. Lastly we require the Killing vector to be null or timelike for all values of $v$; this eliminates the more general cases, leaving only Killing vectors whose component along the null vector $n$ is constant. We use this new result to find timelike or null Killing vectors in a $CSI$ $CCNV$ spacetime.

\end{section}

\begin{section}{$CSI$, $VSI$ and $CCNV$ spacetimes}

\begin{subsection}{Kundt Metrics and $CCNV$ spacetimes} \label{KundtCCNVsect}

Given a spacetime possessing a covariantly constant null covector $\ell_a$, it must be geodesic, non-expanding, shear-free and non-twisting

\beq &\ell^\alpha \ell_{\beta;\alpha} = \ell^\alpha_{~;\alpha} = \ell^{\alpha;\beta}\ell_{(\alpha;\beta)} = \ell^{\alpha;\beta}\ell_{[\alpha;\beta]} = 0. & \label{Kundtvect} \eeq

The null vector will be surface forming since it is geodesic and twist-free, thus there exists locally an exact null one-form $\ell = du$. Choosing local coordinates $(u,v,x^e)$ where $\ell = du$, we use a coordinate transformation and a null rotation to turn the metric into the form

\begin{equation}
ds^2 = 2 du \left(d v+H(u, v, x^e) du+ \hat W_{ e}(u,v,x^f) dx^e\right) + g_{ef}(u,v, x^g) dx^e dx^f. \nonumber
\end{equation}

\noindent The conditions that $\ell$ is non-expanding and shear-free leads to the constraint that $\tilde{g}_{ef,v} = 0$ and so we are left with

\beq & d s^2=2 du \left(d v+H d u+ \hat W_{ e} d x^e\right)+ g_{ef}(u,x^g) dx^e d x^f. \label{Kundt} & \eeq

\noindent Metrics of this form possess a null vector satisfying \eqref{Kundtvect} and will be called higher-dimensional Kundt metrics since these generalize the four-dimensional Kundt metrics \cite{CSI}.


The coordinate transformations that preserve the Kundt form of the metric are:

\begin{enumerate}

\item{} $(u',v',x'^e)=(u,v,f^e(u;{x}^f))$, and $J^e_{~f}\equiv
\frac{\partial f^e}{\partial x^f}$.

\beq
& H'= H+g_{ef}f^e_{,u}f^f_{,u}- \hat W_f\left(J^{-1}\right)^f_{~e}f^e_{,u} ~~~  \hat W'_e =  \hat W_f\left(J^{-1}\right)^f_{~e}-g_{ef}f^f_{,u}& \nonumber \\
&\tilde{g}'_{ef} = \tilde{g}_{hl}\left(J^{-1}\right)^h_{~e}\left(J^{-1}\right)^l_{~f} &\nonumber \eeq

\item{} $(u',v',x'^e)=(u,v+h(u,x^g),x^e)$

\beq H'= H-h_{,u} ~~~  \hat W'_e =  \hat W_e-h_{,e} ~~~ g_{ef}' = g_{ef} \nonumber \eeq

\item{} $(u',v',x'^e)=(g(u), v/g_{,u}(u),x^e)$

\beq H' = \frac 1{g_{,u}^2}\left(H+v\frac{g_{,uu}}{g_{,u}}\right) ~~~  \hat W'_e = \frac 1{g_{,u}} \hat W_e ~~~
\tilde{g}'_{ef} = g_{ef}.\nonumber \eeq

\end{enumerate}

Choosing the coframe $\{ m^a \}$

\beq m^1 = n = dv + Hdu +  \hat W_e dx^e, ~~~ m^2 = \ell ~~~ m^i = m^i_{~e} dx^e \eeq

\noindent where  $m^i_{~e} m_{if} = g_{ef}$ and $m_{ie}m_j^{~e} = \delta_{ij}$, we may rewrite the connection coefficients $L_{ab} = \ell_{(a;b)}$ as

\beq \ell_{(\alpha;\beta)} = L_{11} \ell_{\alpha} \ell_{\beta} + L_{1i} (\ell_{\alpha} m^i_{~\beta} + m^i_{~\alpha} \ell_{\beta}) \label{spincoef}. \eeq

\noindent Since $\ell$ is covariantly constant it is immediately a Killing vector and so  $L_{11}$ and $ L_{1i} $ vanish. Setting $L_{11}$ to zero we find $H_{,v} = 0 .$ In order to find constraints from $L_{1i} = 0$ we must assume    that $m^i_{~e}$ is upper-triangular - this is always possible in thanks to the QR decomposition. Setting $L_{13} = 0$ we find that $ \hat W_{3,v} = 0$. Using this we can inductively show that $\hat W_{i,v} = 0 .$ Thus we find that in order for $\ell$ to be a Killing vector the metric must be independent of the light-cone coordinate $v$, and our metric takes the form
%
\begin{equation}
ds^2=2 du (d v+H(u,x^e)d u+ \hat W_{ e}(u,x^f)d x^e)+ {g}_{ef}(u,x^g) dx^e dx^f.  \label{CCNVKundt}
\end{equation}

\noindent Due to the $v$-independence, $\ell = \partial_v$ must be a covariantly constant null vector.

In general, if there exists a null vector $\ell$ that is a $CCNV$ then from the Ricci identity $ \ell^\alpha R_{\alpha \beta \gamma \delta} = 0. $ The Riemann tensor must be of type II or less and have the following boost-weight decomposition in the aforementioned coframe

\beq R_{\alpha\beta\gamma\delta} &=& R_{ijkl}m^{i}_{ \{ \alpha}m^{j}_{ \beta}m^{k}_{ \gamma}m^{l}_{ \delta\}} +
4R_{2jkl}\ell_{\{ \alpha}m^{j}_{ \beta}m^{k}_{ \gamma}m^{l}_{ \delta \}} \nonumber \\
&~&+  4R_{2j2l}\ell_{ \{\alpha}m^{j}_{ \beta}\ell_{\gamma}m^{l}_{ \delta\}}  .   \label{Riemdecomp} \eeq

\end{subsection}

\begin{subsection}{$CSI$  Spacetimes}

We call a Lorentzian manifold for which all scalar curvature invariants constructed from the Riemann tensor and its covariant derivatives are constant a $CSI$ spacetime, while those Lorentzian manifolds with vanishing scalar curvature invariants will be called $VSI$ spacetimes.  Clearly $VSI \subset CSI$, and the set of all locally homogeneous spacetimes, H, is another subset of the $CSI$ spacetimes.

Let us denote by ${CSI_R}$ the set of all reducible $CSI$ spacetimes that can be built from $VSI$ and $H$ by (i) warped products (ii) fibered products, and (iii) tensor sums. Similarly we denote by ${CSI_F}$ those spacetimes for which there exists a frame with a null vector $\ell$ such that all components of the Riemann tensor and its covariant derivatives in this frame have the property that (i) all positive boost weight components (with respect to $\ell$) are zero and (ii) all zero boost weight components are constant. Finally, we denote by ${CSI_K}$, those $CSI$  spacetimes that belong to the (higher-dimensional) Kundt class \eqref{Kundt}, the so-called Kundt $CSI$  spacetimes. By construction ${CSI_R}$ is at least of Weyl type $II$, and by definition ${CSI_F}$ and ${CSI_K}$ are at least of Weyl type $II$.

In the Riemannian case, a manifold with constant scalar invariants is immediately homogeneous, $(CSI \equiv H)$. This is not true for Lorentzian manifolds; we may find examples of $CSI$  spacetimes which consists of combinations of  homogeneous spaces and a certain subclass of Kundt spacetimes in \cite{CSI}. Interestingly, for every $CSI$ spacetime there is a homogeneous spacetime (not necessarily unique) with precisely the same constant invariants.  This suggests that $CSI$  spacetimes can be constructed from homogeneous spacetimes $H$ and $VSI$ spacetimes using warped and fibered products (e.g., ${CSI_R}$).

In \cite{CSI} it is conjectured that a spacetime is $CSI$, if and only if, there exists a null frame in which the Riemann tensor and its derivatives can be brought into one of the following forms: either

\begin{enumerate}
\item{} the Riemann tensor and its derivatives are constant, in which case we have a locally homogeneous space, or \item{} the Riemann tensor and its derivatives are of boost order zero with constant boost weight zero components at each order.  This implies that Riemann tensor is of type II or less.
\end{enumerate}

\noindent Assuming that the above conjecture is correct and that there exists such a preferred null frame, it was then conjectured that if a spacetime is $CSI$, then the spacetime is either locally homogeneous or belongs to the higher dimensional Kundt $CSI$  class, and if a spacetime is $CSI$ it can be constructed from locally homogeneous spaces and $VSI$ spacetimes. This construction can be done by means of fibering, warping and tensor sums. From the results above and these conjectures, it is plausible that for $CSI$  spacetimes that are not locally homogeneous, the Weyl type is $II$, $III$, $N$ or $O$, and that all boost weight zero terms are constant.

\end{subsection}

\begin{subsection}{$VSI$ Spacetimes}

There is more known about the subclass of $CSI$  spacetimes in which all curvature invariants vanish, the $VSI$ spacetimes. From \cite{Higher} we have the following theorem:

\begin{thm}[$VSI$ Theorem] All curvature invariants of all orders vanish in an N-dimensional Lorentzian spacetime if and only if there exists an aligned, non-expanding, non-twisting, shear-free, geodesic null direction $\ell^\al$ along which the Riemann tensor has negative boost order.
\end{thm}

\noindent To be precise, there exists a null vector $\ell$ such that

\beq \ell_{\al;\bt} = L_{11} \ell_\al \ell_\bt + L_{1i} \ell_\al m^i_{~\bt} + L_{i1} m^i_{~\al} \ell_\bt. \nonumber  \eeq

If we choose a frame including $\ell$ as a basis vector, the Riemann tensor will be of type III, N or O, this implies that all $VSI$ spacetimes belong to the generalized Kundt class. In fact it was shown in \cite{CFHP} that there is a canonical form for the metrics of $VSI$ spacetimes:

\beq & d s^2=2 du (dv+H du +  \hat W_{ e} dx^e)+ \delta_{ef} dx^e dx^f, \label{VSIKundt} & \eeq

\noindent where the metric functions satisfy certain constraints that will be discussed in the next subsection. The authors of \cite{CFHP} go further by classifying all higher-dimensional $VSI$ spacetimes according to their Weyl type, Ricci-type and whether $\hat{W}_3$ has $v$-dependence or not.

\end{subsection}

\begin{subsection}{Supergravity and $CSI$ Spacetimes}

A subset of Ricci type N $VSI$ spacetimes, the higher-dimensional pp-wave spacetimes, have been studied  in the literature, and are known to be exact solutions in string theory \cite{string, hortseyt, tseytlin}, in type IIB supergravity with an R-R five-form \cite{matsaev},  and with  NS-NS form fields as well \cite{RT}. The pp-wave spacetimes are of Weyl type N. However there are Ricci type N solutions of Weyl type III, like the string gyratons and in fact all Ricci type N $VSI$ spacetimes are solutions to supergravity \cite{VSISUG}.

Moreover in \cite{VSISUG} it was shown that there are $VSI$ spacetime solutions of type IIB supergravity which are of Ricci type III, assuming appropriate source fields are provided. In order for a $VSI$ spacetime to be of Ricci type III, the dilaton, $\phi$, must be a function of $u$ and the metric will have $v$-dependence. However, no null or timelike Killing vectors can exist in a $VSI$ spacetime if the metric is dependent on $v$ \cite{VSISUG}, thus the Ricci type III spacetimes do not preserve supersymmetry. While the Ricci type N, Weyl type III solutions can be reduced to Weyl type N, the Ricci type III solution can only have Weyl type III.

In four dimensions $VSI$ spacetimes are known to be exact string solutions to all orders in the string tension $\alpha '$, even in the presence of additional fields \cite{coley}. Similarly it can be shown that higher-dimensional supergravity solutions supported by the proper fields (for example, the dilaton scalar field, Kalb-Ramond field, and form fields) are also exact solutions in string theory using arguments from \cite{string,matsaev,RT}. Thus it can be analogously argued that the $VSI$ supergravity spacetimes are exact string solutions to all orders in the string tension $\alpha '$ in the presence of the appropriate fields, and so it is to be expected that the special $VSI$ supergravity solution introduced in \cite{VSISUG} will be as well. From this we conclude that the Ricci type III solution may be of relevance to string theory.

Since the Ricci type III IIB-supergravity solutions do not preserve supersymmetry, this leads to the question of what are the necessary conditions to preserve supersymmetry? In a number of supergravity theories (e.g. D= 11 \cite{hommth}, type IIB \cite{iibkv}), in order to preserve some supersymmetry, it is necessary that the spacetime admits a Killing spinor $\epsilon^A$ which then yields a  null or timelike Killing vector from its Dirac current $x^\al = \bar \epsilon^{A'} \gamma^\al_{AA'} \epsilon^A$ where the $\gamma^{\al}_{~AA'}$ are the higher dimensional analogues of the gamma matrices.

In the case of N-dimensional $VSI$ spacetimes, the existence of Killing vectors depends on whether the components of the metric are independent of the light-cone coordinate $v$. This requirement leads to the conclusion that all $VSI$ spacetime solutions to type IIB supergravity preserving some supersymmetry are of Ricci type N, Weyl type III(a) or N \cite{CFHP}. Such spacetimes include not only pp-waves but also spacetimes of Weyl type III(a), an example of which is the string gyratons \cite{strgyro}. Weyl type III(a) spacetimes like the vacuum solution or NS-NS solutions only preserve some of the supersymmetry.

It is known that $AdS_d \times S^{(D-d)}$ (in short $AdS\times S$) is an exact solution of supergravity (and preserves the maximal number of supersymmetries for certain values of $(D,d)$ and for particular ratios of the radii of curvature of the two space forms. Such spacetimes  (with $d=5,D=10$) are supersymmetric solutions of IIB supergravity (there are analogous solutions in $D=11$ supergravity) \cite{kumar}: $AdS \times S$ is an example of a $CSI$ spacetime \cite{CSI}. There are a number of other $CSI$ spacetimes known to be solutions of supergravity and admit supersymmetries; namely, generalizations of $AdS \times S$ (for example, see \cite{Gauntlett}) and (generalizations of) the chiral null models \cite{hortseyt}. The $AdS$ gyraton (which is a $CSI$ spacetime with the same curvature invariants as pure AdS) \cite{FZ} is a solution of gauged supergravity  \cite{Caldarelli} (the $AdS$ gyraton can be cast in the Kundt form \cite{CFH}).

Other known $CSI$ spacetimes have been investigated as solutions of supergravity. For example, we can consider the product manifolds of the form $M \times K$, where $M$ is an Einstein space with negative constant curvature and $K$ is a (compact) Einstein-Sasaki spacetime. The warped product of $AdS_3$ with an 8-dimensional compact (Einstein-Kahler) space $M_8$ with non-vanishing 4-form flux are supersymmetric solutions of D=11 supergravity \cite{Gauntlett}, while in \cite{membranes} supersymmetric solutions of $D=11$ supergravity, where $M$ is the squashed $S^7$, were given.

A class of $CSI$ spacetimes which are solutions of supergravity and preserve supersymmetries, were built from a  $VSI$  seed and locally homogeneous (Einstein) spaces by warped products, fibered products and tensor sums \cite{CSI}, yielding generalizations of $AdS\times S$ or $AdS$ gyratons  \cite{CFH}.  In particular, solutions obtained by restricting attention to  $CCNV$  and Ricci type N spacetimes were considered, some explicit examples of $CSI$ supergravity spacetimes were constructed by taking a homogeneous (Einstein) spacetime, $(\mathcal{M}_{\text{Hom}},\tilde{g})$ of Kundt form and generalizing to inhomogeneous spacetimes, $(\mathcal{M},{g})$ by including arbitrary Kundt metric functions (by construction, the curvature invariants of $(\mathcal{M},{g})$ will be identical to those of $(\mathcal{M}_{\text{Hom}},\tilde{g})$); a number of 5D examples were given, in which $ds_{hom}^2$ was taken to be  Euclidean space or hyperbolic space \cite{cardoso}.

\end{subsection}

\end{section}

\begin{section}{Curvature}

\begin{subsection}{Riemann tensor components}

As a note the frame indices $a,b,c,d$ range from 1 to $N$, $i,j,k,l,h$ lie between 3 and $N$ and $m,n,p,q \in [4,N]$. Equivalently the coordinate indices $\alpha ,\beta ,\gamma, \delta \in [1,N]$, $e,f,g \in [3,N]$ and $r,s,t,v \in [4,N]$.

The Riemann tensor components of the metric \eqref{CCNVKundt} were found in \cite{mcnutt} using the Cartan Structure equations. To summarize we first choose the following null coframe

\beq \omega^1 &=& \ell = du. \nonumber \\
\omega^2 &=& n = dv + H du +  \hat W_e dx^e. \label{nullframe} \\
\omega^i &=& m^i = m^i_{~e} dx^e,  ~~~~m^i_{e}m_{if} = \tilde g_{ef}.  \nonumber \eeq

\noindent We then defined the matrix-valued connection-one form, $\Gamma^{ a  }_{~b} = \Gamma^{ a  }_{~b c } \omega^{c }$, satisfying the Cartan Structure equations

\beq & d\omega^{ a  } = - \Gamma^{ a  }_{~b }\wedge \omega^{ b } = \Gamma^{ a  }_{~b c }\omega^{b } \wedge \omega^{c},& \label{cartane1}\\
&\Theta^{a}_{~b} = d\Gamma^{a}_{~b}+\Gamma^{a}_{~c} \wedge \Gamma^{c}_{~b} = \frac{1}{2}R^{a}_{~bcd} w^{c} \wedge w^{d}.& \label{cartane2} \eeq

\noindent Exterior differentiation of the above coframe gives

\beq d \bfn &=& (H_{,e} -  \hat W_{e,u})m_{j}^{~e} m^{j} \wedge \ell + 2 \hat W_{[e,f]}m_{i}^{~e}m_{j}^{~f} m^{j} \wedge m^{i} \nonumber \\
d\ell &=& 0 \\
dm^{i} &=& m^{i}_{~e,u}m_{j}^{~e} \ell \wedge m^{j} + m^i_{~e,f}m_{[j}^{~e}m_{k]}^{~f}m^{k} \wedge m^{j}. \nonumber  \eeq

\noindent We compare this with the first Cartan Structure equation and solve the resulting system of linear equations involving the connection components to find that:

\beq  \Gamma_{2i2} &=& J_i =  (H_{,e}- \hat W_{e,u})m_{i}^{~e}, \\
 \Gamma_{2ij} &=& - \frac 12 A_{ij} - B_{(ij)}, \label{G2ij}\\
 \Gamma_{ij2} &=& \frac 12 A_{ij} - B_{[i j]}, \label{otherc}\\
 \Gamma_{ijk} &=& - \frac{1}{2}[D_{ijk}+D_{jki}+D_{kji}] \label{gammaijk} \eeq

\noindent where, the above tensors, $A,B,D$ are defined to be:

\beq  D_{ijk} &=& 2m_{i e,f}m_{[j}^{~e}m_{k]}^{~f},  \label{Dijk} \\
 A_{ij} &=& 2 \hat W_{[e,f]}m_{i}^{~e}m_{j}^{~f}, \label{Aij}  \\
 B_{ij} &=& m_{i e,u}m_{j}^{~e} \label{Bij}.  \eeq

\noindent From the second Cartan Structure equation, we have that the non-zero components of the Riemann tensor are $R_{2ij2}$ $R_{2ijk}$ and $R_{ijkl}$, displayed below in terms of $A$, $B$, $D$ and the connection components of the transverse metric:

\beq
R_{2 i j 2}&=& J_{i,e}m_{j}^{~e} + [  \frac 12 A_{ij} + B_{(ij)} ]_{,u} + [ \frac 12 A_{ik} + B_{(ik)} ]  B^{k}_{~j} \\
 & & - J_{k} \Gamma ^{k}_{~ i j} + \displaystyle \sum_k [ \frac{1}{2}A_{kj} + B_{(kj)}][ \frac{1}{2}A_{ki} - B_{[ki]}]  \nonumber\\
\nonumber \\
R_{2 i j k}&=& -[ \frac{1}{2} A_{i k} + B_{(i k)} ]_{,e}m_{j}^{~e} + [ \frac 12 A_{i j} + B_{(i j)} ]_{,e} m_{k}^{~e} \\
& & + [ A_{i l} + 2B_{(i l)} ] m^l_{~[e,f]}m_j^{~e}m_k^{~f} -  [ \frac{1}{2} A_{l j} + B_{(l j)} ] \Gamma^l_{~ i k} \nonumber\\
& & + [ \frac 12 A_{l k} + B_{(l k)} ] \Gamma^l_{~ i j} \nonumber \\
\label{riembw-1} \\
R_{i j k l}&=& [\Gamma_{i j h} m^{h}_{~e}]_{,f}(m_{k}^{~f}m_{l}^{~e} - m_{l}^{~f}m_{k}^{~e})- \Gamma_{i h l} \Gamma^{ h}_{~ j k} + \Gamma_{i h k} \Gamma^{ h}_{~ j l}. \label{riembw0} \eeq



\end{subsection}

\begin{subsection}{Ricci Tensor Components}

Contracting the first and third indices of the Riemann tensor, we obtain the non-zero components of the Ricci tensor

\beq
R_{22} &=& \displaystyle \sum_{i} [ - J_{i,e}m_{i}^{~e} - (  \{ \frac 12 A_{ii} \} + B_{(ii)} )_{,u} - ( \frac 12 A_{ik} + B_{(ik)} )  B^{ki} \\
 & & + J_{k} \Gamma ^{k}_{~ i i} - ( \frac{1}{2}A_{ki} + B_{(ki)})( \frac{1}{2}A_{ki} - B_{[ki]})]  \nonumber\\
\nonumber \\
R_{2i} &=& \displaystyle \sum_j [( \frac{1}{2} A_{j i} + B_{(j i)} )_{,e}m_{j}^{~e} - (B_{(j j)} )_{,e} m_{i}^{~e} \label{ricbw-1}\\
& & - ( A_{j l} + 2B_{(j l)} ) m^l_{~[e,f]}m_j^{~e}m_i^{~f} +  ( \frac{1}{2} A_{l j} + B_{(l j)} ) \Gamma^l_{~ j i} \nonumber\\
& & -  (\frac 12 A_{l i} + B_{(l i)} ) \Gamma^l_{~ j j}] \nonumber \\
\nonumber \\
R_{ij} &=& \displaystyle \sum_k [ (\Gamma_{k i l} m^{l}_{~e})_{,f}(m_{k}^{~f}m_{j}^{~e} - m_{j}^{~f}m_{k}^{~e})- \Gamma_{k l j} \Gamma_{l i k} + \Gamma_{k l k} \Gamma_{l i j}]. \label{ricbw0}
\eeq



\noindent By \eqref{G2ij} this yields $\displaystyle \sum_i \Gamma_{2ii}=-B_{ii}=-g^{ij}\partial_{u}g_{ij}$, but from the definition of $B_{ij}$ we have the identity $B_{ij}+B_{ji}= -m_{ie}m_{jf}\partial_{u}g^{ef}.$ Consequently we have that

\beq & 2R_{2i}=-(g^{jk}\partial_{u}g_{jk})_{|i}-A_{il|l}+(B_{il}+B_{li})_{|l} \label{ric2l}. & \eeq

\noindent We note that by contracting $B_{ij}+B_{ji}= -m_{ie}m_{jf}\partial_{u}g^{ef}.$ with $m^{i}_{~e}m^{jf}$ gives $B_{i}^{~j}+B^{j}_{~ i} = -g^{ik}\partial_{u}g_{ki}$ in terms of coordinate indices.  This can be used to express the last term of \eqref{ric2l} in terms of a divergence over the coordinate index $j$.
	
\end{subsection}

\end{section}

\begin{section}{Criteria for $CCNV$ spacetimes to have constant curvature invariants} \label{CSICriteria}

It has been shown that the line-element  \eqref{CCNVKundt} for a spacetime admitting a covariantly constant null vector has a Riemann curvature tensor with the following boost weight decomposition:

\beq  R_{\alpha\beta\gamma\delta} &=& \overbrace{R_{ijkl}m^{i}_{\ \{\alpha}m^{j}_{\ \beta}m^{k}_{\ \gamma}m^{l}_{\ \delta\}}}^0  + \overbrace{R_{2jkl}\ell_{\{\alpha}m^{j}_{\ \beta}m^{k}_{\ \gamma}m^{l}_{\ \delta\}}}^{-1} + \\
& & \overbrace{R_{2j2l}\ell_{\{\alpha}m^{j}_{\ \beta}\ell_{\gamma}m^{l}_{\ \delta\}} }^{-2}.  \label{riembwd} \nonumber
\eeq

\noindent We note from \eqref{riembw0} that only the curvature tensor of the transverse metric contributes to the boost-weight 0 components of the curvature tensor.

In order to take the covariant derivative of \eqref{riembwd}, we first consider the covariant derivative of the frame components $R_{abcd}$:

\beq
\nabla_{\epsilon}R_{abcd}=\ell_{~\epsilon}D_{2}R_{abcd}+m^{i}_{\ \epsilon}D_{i}R_{abcd}  \label{nabfriem}
\eeq

\noindent where $D_a$ denotes the derivative of a function with respect to the frame vectors $\{ m_i \}$ and we have used the $v$-independence of \eqref{CCNVKundt} to set $D_{1}R_{abcd}=0$. The covariant derivatives of the frame one-form, $\ell$ is $\nabla_{\epsilon}\ell_{\alpha}=0$, while for the remaining frame one-forms $m^i$,

\beq \nabla_{\epsilon}m_{i\alpha}=\Gamma_{lin}m^{l}_{~\alpha}m^{n}_{\ \epsilon} +\Gamma_{li2}m^{l}_{~\alpha}\ell_{\epsilon}+\Gamma_{2il}\ell_{\alpha}m^{l}_{~\epsilon}+\Gamma_{2i2}\ell_{\alpha}\ell_{\epsilon} \, . \label{mideriv} \eeq

Thus \eqref{nabfriem} and \eqref{mideriv} imply that $\nabla$ does not raise boost weight because covariant differentiation does not introduce the null vector $\bfn$ into the expressions; as a result $\nabla_{\epsilon}R_{\alpha\beta\gamma\delta}$ will contain frame components whose highest boost weight is 0 and these will correspond only to the covariant derivative of the curvature of the Riemannian transverse space.  Using an inductive argument, this can be shown to hold for any number of covariant derivatives of the Riemann tensor. Let the $k^{th}$ covariant derivative of \eqref{riembwd} be represented symbolically as $\nabla^{k}R$, then we can say that $\nabla^{k}R$ has frame components whose highest boost weight is zero and $\ell$ contracted with any index of $\nabla^{k}R$ vanishes; i.e., $\nabla^{k}R \cdot \ell = 0$.

It now follows that all curvature invariants of \eqref{CCNVKundt} will be completely equivalent to the curvature invariants of the transverse space.  Therefore, if we impose the $CSI$ condition on \eqref{CCNVKundt}, we are requiring the transverse space to be a $CSI$ Riemannian metric.  From a theorem of \cite{prufer} we conclude that the transverse metric is locally homogeneous, which establishes the following,

\begin{lem}\label{lem:ccnvcsi}
A generalized Kundt metric admitting a  $CCNV$  is $CSI$ if and only if the transverse metric is locally homogeneous.
\end{lem}

Now, consider the Ricci invariant $r_{2}=R_{ab}R^{ab}=R_{ij}R^{ij}$, where the second equality follows from the form of the Ricci tensor \eqref{ricbw-1}-\eqref{ricbw0}, which shows that the boost weight 0 components arise solely from the transverse metric.  Since $r_{2}=\sum_{i,j}(R_{ij})^2$ is a sum of squares, we have that if $r_{2}=0$ then $R_{ij}=0$.  A theorem of \cite{aleksee} states that a homogeneous Riemannian space that is Ricci-flat is flat.  Therefore, combining these with Lemma \eqref{lem:ccnvcsi} gives the result:

\begin{prop}\label{prop:ccnvvsi}
If a generalized Kundt metric admitting a  $CCNV$  is $CSI$ and $R_{ab}R^{ab}=0$ then the metric is $VSI$.
\end{prop}

Although $R_{\alpha\beta\gamma\delta}\ell^{\alpha}=0=R_{\alpha\beta}\ell^{\alpha}$, it does not follow that $C_{\alpha\beta\gamma\delta}\ell^{\alpha}=0$.  More precisely, if we consider the decomposition of the Riemann tensor into its trace and trace-free parts we obtain

\beq
&C_{1bcd}=\frac{1}{(N-2)}(\eta_{1d}R_{bc}-\eta_{1c}R_{bd})+\frac{1}{(N-1)(N-2)}(\eta_{1c}\eta_{bd}-\eta_{1d}\eta_{bc})R \, .& \label{1weyl}
\eeq

\noindent It is clear that there will exist boost weight 0 and -1 components of Weyl with projections along $\ell$ that do not vanish, namely

\begin{eqnarray}\label{1weylcmpts}
&C_{1212} = \displaystyle{\frac{-R}{(N-1)(N-2)}}, ~~ C_{1i2j} = \displaystyle{\frac{-R_{ij}}{N-2}} + \displaystyle{\frac{R\delta_{ij}}{6}}, ~~ C_{12i2} = \displaystyle{\frac{R_{2i}}{N-2}}.& \,
\end{eqnarray}

\noindent Assuming the conditions in proposition \eqref{prop:ccnvvsi} are satisfied then the Weyl components in \eqref{1weylcmpts} vanish and furthermore $C_{ijkl}=0$ as well \cite{aleksee}. The remaining non-vanishing Weyl components of the $VSI$ metric are $C_{2jkl}$ and $C_{2j2l}$.

\end{section}

\begin{section}{$CCNV$ spacetimes admitting an additional isometry}  \label{GeneralCCNV}

\noindent The Killing equations for $X = X_1 n + X_2 \ell + X_i m^i$ are:

\beq &D_{1}X_{1} = 0 \label{killeqn0}\\
&D_{2} X_{1} + D_{1} X_{2} = 0& \label{killeqn1}\\
&D_{3}X_{1} + D_{1} X_{3} = 0& \label{killeqn2}\\
&D_{m} X_{1} = 0& \label{killeqn3}\\
&D_2 X_2 + \displaystyle \sum_i  J_i X_i=0 & \label{killeqn4}\\
& D_i X_2 + D_2 X_i - J_i X_1- \displaystyle \sum_j (A_{ji}+B_{ij})X_j = 0 & \label{killeqn5} \\
& D_j X_i + D_i X_j +2B_{(ij)}X_1 - 2\displaystyle \sum_k \Gamma_{k(ij)} X_k = 0& \label{killeqn6}. \eeq

To start, we make a coordinate transformation to eliminate $ \hat W_3$ in \eqref{CCNVKundt}

\beq (u',v',x'^i) = (u, v+h(u,x^k),x^i), ~~~ h = \int  \hat W_3 dx^3. \label{now3} \eeq

\noindent This choice of coordinates will be useful in order to write down the metric functions and Killing covector components themselves or determining equations for them. For any point in the manifold we may rotate the frame, setting $X_3 \neq 0$ and $X_m = 0$.
This can be done by taking the spatial part of the Killing form $X = X_1n + X_2 \ell + X_i  m^i$ and choosing

\beq m^3 = \frac{1}{\chi}X_im^i~~~~\chi = \sqrt{\displaystyle \sum_i X_i^{~2}}. \eeq

\noindent Using Gram-Schmidt orthonormalization it is possible to determine the remaining vectors for the frame basis. This is a local orthogonal rotation so the form of our metric remains unchanged while $X$ is now

\beq X = X_1 n + X_2 \ell + \chi m^3 . \eeq

\noindent Henceforth it will be assumed that the matrix $m_{ie}$ is upper-triangular, due to the QR decomposition.

The frame derivatives are

\beq \ell ~=~ D_1 &=&  \partial_v \nonumber \\
n ~=~ D_2  &=&  \partial_u - H\partial_v \label{framed}\\
m_i ~=~ D_i &=&  m_i^{~e}(\partial_{e} -  \hat W_e \partial_v) \nonumber \eeq

\noindent Thus \eqref{killeqn0} -- \eqref{killeqn3} imply that the Killing vector components are of the form:

\beq  & X_1 = F_1(u,x^e) &  \nonumber \\
&X_2 = -D_2(X_1)v+F_2(u,x^e)& \label{kvcomps} \\
&X_3 = -D_3(X_1)v +F_3(u,x^e). & \nonumber \eeq

\noindent The remaining Killing equations (\eqref{killeqn4}- \eqref{killeqn6} involve $A_{ij}$ and $J_i$, if we define $W_i = m_i^{~e} \hat W_e$, we may write  $\Gamma_{2i2} = J_i$ and $A_{ij}$ in terms of frame derivatives

\beq  J_i &=& D_i H - D_2 W_i - B_{ji}W^j \label{Jieqn} \\
 A_{ij} &=& D_{[j} W_{i]}  + D_{k[ij]}W^k. \label{Aijeqn} \eeq

From the commutation relations
 
\beq \mbox{} [D_{1},D_{a}] & = & 0 \nonumber \\
\mbox{} [D_{2},D_{j}] & = & J_{j}D_{1}-\sum_{i}B_{ij}D_{i} \label{com} \\
\mbox{} [D_{k},D_{j}] & = & A_{kj}D_{1}+2\sum_{i}\Gamma_{i[kj]}D_{i} \, . \nonumber \eeq

\noindent applied to the Killing equations the following can be derived:  

\beq D_aD_3X_1 = \Gamma_{3na}D_3X_1 = 0, a = 2,3,...N \label{ComKv} \eeq


\noindent This leads to two cases either $D_3 X_1 = 0$ or $\Gamma_{3n2} = \Gamma_{3n3} = \Gamma_{3nm} = 0$ and $X_1$ is  linear in $x^3$. Supposing there exists another Killing vector $X$ we will find further constraints on its components $X_a$ as well as the metric functions $W_e$ and $m_{ie}$ in the ensuing subcases.

\begin{subsection}{Case 1: $D_{3}X_{1}=0$}

Equation \eqref{killeqn3} implies $X_1$ must be independent of all spacelike coordinates. Using equation \eqref{killeqn4} and the definition of $F_2$ from \eqref{kvcomps}, we have that $X_1$ must be of the form

\beq X_{1}=c_1u+c_2. \eeq

\noindent If $c_1 \neq 0$ we may always use a type (3) coordinate transform from Section \ref{KundtCCNVsect} to set $X_1 = u$, while if $c_1 = 0$ we may choose $c_2 = 1$ by scaling all coordinates by $c_2$ in both cases the functions $F_2$, $F_3$, $H$ and $W_e$ in the new coordinate system are just the original functions multiplied by constants.

Equations \eqref{ComKv} are identically satisfied, and \eqref{killeqn4}-(\ref{killeqn6}) reduce to:

\begin{eqnarray}
c_1 H + D_2 F_{2} + J_3 F_3 & = & 0 \label{1bcase1} \\
- J_3 X_1 + D_3 F_2 + D_2 F_3 - B_{33} F_3 & = & 0 \label{2b1case1} \\
c_1 W_n - J_n X_1 + D_n F_2 - ( A_{3n} + B_{n3} ) F_3 & = & 0 \label{2b2case1} \\
B_{33} X_1 + D_3 F_3 & = & 0 \label{3b1case1} \\
2 B_{(3n)} X_1 + D_n F_3 - \Gamma_{3n3} F_3 & = & 0 \label{3b2case1} \\
2 B_{(nm)} X_1 - 2 \Gamma_{3(nm)} F_3 & = & 0  \, . \label{3b3case1}
\end{eqnarray}

\noindent Setting $c_{1}=c_{2}=F_{3}=0$, $X$ reduces to a scalar multiple of the known Killing covector $\ell$. We must consider the possibility where $F_3$ vanishes.

\end{subsection}

\begin{subsection}{Subcase 1.1: $F_{3}=0$}



\noindent Setting $F_3 = 0$ in equations \eqref{3b1case1}-\eqref{3b3case1} imply that $B_{(ij)}=0$. Rewriting  this as
$B_{(ij)} = m_i^{~e}m_j^{~f}g_{ef,u}$, the metric is independent of $u$. By virtue of the upper-triangular form of $m_{ie}$ we see it must be independent of $u$ also. 

Assuming $c_1 \neq 0$, we make the appropriate coordinate transformation to set $X_1 = u$, equation \eqref{1bcase1} yields $H$ algebraically:

\beq H = - D_2 F_2. \nonumber  \eeq

\noindent Solving the resulting differential equation from \eqref{2b2case1}, $W_m$ is expressed as:

\begin{eqnarray}
W_m  & = & \frac{1}{u}\left[\int -D_{m}(u D_2 F_2 + F_2) du + B_{m}(x^e)\right]. \nonumber
\end{eqnarray}

\noindent Taking \eqref{2b1case1} with $J_3 = D_3 H$ we see that

\beq D_2 D_3 ( u F_2) = 0,  \nonumber \eeq

\noindent implying that $F_2$ must be of the form

\beq F_2 = \frac{f_2(x^e)}{u} + \frac{g_2(u)}{u}. \label{case11f2}. \eeq

\noindent We rewrite the equations of $H$ and $W_m$ in terms of these two functions

\beq
H &=&  \frac{ f_2(x^e) }{ u^2 } - \frac{ g_2'(u) }{ u } +  \frac{ g_2(u) }{ u^2 } \label{C11H} \\
W_m &=& \frac{ B_{m}(x^e) }{ u } \label{C11Wn} \eeq

\noindent where $g' $ denotes the derivative of $g$ with respect to $u$

If $c_1 = 0$, $F_2$ must be independent of $u$, we rescale our coordinates so that $X_1 = 1$, the equations for $H$ and $W_n$ are

\beq H &=&  F_2(x^e) + A_0(u, x^r) \label{C11H0} \\
W_n &=&  \int D_nA_0 du + C_n(x^e). \label{C11Wn0} \eeq


\noindent In either case, the only requirement on the transverse metric is that it be independent of $u$. The arbitrary functions in this case are $F_{2}$ and the functions arising from integration.

\end{subsection}

\begin{subsection}{Subcase 1.2: $F_{3} \neq 0$}

As a consequence of the upper triangular form of $m_{ie}$ the system of equations \eqref{1bcase1} -- \eqref{3b3case1} decouples in the following order.  Beginning with equation \eqref{3b1case1}, we may reduce  this to an  equation for $m_{33}$ in terms of $F_{3}$

\begin{equation}
\frac{m_{33,u}}{m_{33}}= -\frac{1}{X_{1}}D_{3}F_{3},  \label{m33,uode}
\end{equation}

\noindent whose solution is

\beq m_{33} = -\int \frac{1}{X_1}F_{3,3} du + A_1(x^3, x^r). \label{m33,u} \eeq

\noindent Next, consider the diagonal components of \eqref{3b3case1} followed by the off-diagonal components to find the determining equation for $m_{nr}$

\begin{equation}
m_{nr,u}= -m_{nr,3} \frac{F_3}{m_{33} X_{1}},   \label{mnr,u}
\end{equation}

\noindent while equation (\ref{3b2case1}) reduces to

\begin{equation} m_{3r,u}=  - \frac{ F_{3,r}}{X_1} - \frac{ m_{3[r,3]}m_3^{~3}F_3}{X_1}. \label{m3r,u} \end{equation}

With the transverse metric now determined and assuming $c_1 \neq 0$,  we again choose coordinates so that $u' = c_1 u + c_2$, equations \eqref{1bcase1} and \eqref{2b1case1} lead to the form of $H$

\beq H = - D_2F_2 - \frac{ D_2(F_3^2) }{ 2u } - \frac{ F_3 D_3F_2  }{ u } - \frac{ F_3 D_3(F_3^2)}{ 2u^2 }. \label{C12H} \eeq

\noindent The form of $A_{ij}$ expressed in frame derivatives \eqref{Aijeqn}, along with equations \eqref{3b2case1} and \eqref{3b3case1} simplify \eqref{2b2case1} to become the determining equation for the $W_n$

\begin{equation}
D_{2}( u W_n ) + F_3 D_3 W_n + D_n(F_2 - u H ) = 0 \, . \label{C12Wn}
\end{equation}

\noindent Given $F_2(u,x^e)$ and $F_3(u, x^e)$, we treat the equations \eqref{C12H} and \eqref{C12Wn} as constraining equations for $H$ and the $W_n$.

If $c_1$ vanishes, rescale to make $c_2 = 1$,  from \eqref{1bcase1} and \eqref{2b1case1} $F_2$ satisfies the equation 

\beq D_2F_2 + F_3 D_3F_2 + \frac12 D_2(F_3^2) + \frac12 F_3 D_3(F_3^2) = 0. \label{C12f20} \eeq

\noindent The metric function $H$ may be written as

\beq H &=& \int m_{33} D_2F_3 dx^3 + F_2 + \frac12 F_3^2 + A_2(u, x^r). \label{C12H0} \eeq


\noindent The only equation for $W_n$ is

\beq & F_3 D_3W_n + D_2 W_n  = D_n(H) . \label{C12Wn0} & \eeq

If $c_2 \neq 0$ the equation for $m_{ie}$ holds, however if $X_1= 0$ they simplify. The equations \eqref{m33,uode} and \eqref{mnr,u} become

\beq F_{3,3} =0,  \label{f3,u} \\
m_{nr,3} = 0.   \label{mnr,u0} \eeq

\noindent A constraint on the function $m_{33}$ arises from equation \eqref{2b1case1}

\beq D_2log(m_{33}) = - \frac{ D_3F_2 }{ F_3 } - D_2log(F_3). \label{m33,u0} \eeq	

\noindent From \eqref{2b2case1} $W_n$ is found

\beq W_n = - \int \frac{ m_{33} D_nF_2}{F_3} dx^3 + E_{n}(u, x^r), \label{C12Wn00} \eeq


\noindent while equation \eqref{1bcase1} gives $H$

\beq H = - \int \frac{ m_{33} D_2F_2}{F_3} dx^3 + A_3(u, x^r). \label{C12H00} \eeq


There are two further subcases to consider, expanding and simplifying equation \eqref{3b2case1}

\beq \frac{m_{3r,3}}{F_3} =  \left( \frac{m_{33}}{F_3} \right)_{,r}  \label{C12splitter} \eeq

\noindent and so, either $m_{3r}$ is a function of $x^3$ or not. If $m_{33,r} \neq 0$ we may integrate \eqref{C12splitter} for $m_{3r}$

\beq m_{3r} = \int \left( \frac{m_{33}}{F_3} \right)_{,r} F_3 dx^3 + G_{r}(u, x^s). \label{m3r,3} \eeq

\noindent Thus the above along with \eqref{f3,u} and the requirement that $m_{mr,3} = 0$ are the only conditions on the matrix $m_{ie}$. If $m_{3r}$ is independent of $x^3$ then we must have $D_n \left( \frac{ m_{33} }{ F_3 } \right)= 0. $ implying

\beq m_{33,r} = 0. \label{ola} \eeq

\noindent Substituting $m_{33}(u, x^3)$ into \eqref{m33,u0} yields a differential equation, whose solution is

\beq m_{33} = \frac{ 1 }{ F_3 } \int -F_{2,3} du + A_4(x^3). \label{m33c121} \eeq


\end{subsection}

\begin{subsection}{Case 2: $\Gamma_{3ia} = 0$} \label{GeneralCCNVc2}

To investigate what constraints these requirements give, we expand the expressions for the connection coefficients in question:

\beq  \Gamma_{3n2} &=& W_{[3,f]}m_3^{~3}m_n^{~f} - m_{3e,u}m_n^{~e} + m_{n3,u}m_3^{~3}, \nonumber \\
 \Gamma_{3ni} &=& \frac{-1}{2} m_{3[e,f]}m_n^{~e}m_i^{~f} + m_{n[f,3]}m_i^{~f}m_3^{~3} + m_{i[f,3]}m_n^{~f}m_3^{~3} . \nonumber  \eeq

\noindent These constraints lead to the following facts about the metric functions

\begin{lem}\label{lem:metricfn1}

The vanishing of  $\Gamma_{3ia}$ imply the metric functions of \eqref{CCNVKundt} must satisfy the following constraints

\beq \hat W_{[3,r]} = m_{3[3}m^3_{~~r],u} \label{C2Wcon1} \\
m_{3[3,r]} = 0 \label{C2m3rcon1}\\
m_{3[r,s]} = 0 \label{C2m3rcon2} \\	
g_{rs,3} = 0. \label{gcon1} \eeq

\end{lem}

\begin{proof}

To begin the proof consider $\Gamma_{3n2}$, using the upper-triangular form for $m_{ie}$ this simplifies to be

\beq \hat W_{[3,r]}m_3^{~3}m_n^{~r} - m_{3e,u}m_n^{~~e} = 0. \nonumber \eeq

\noindent Multiplying through by $m^n_{~~s}$, we note that $m_n^{~~e}m^n_{~~s} = \delta^e_{~~s} - m_3^{~~e}m^3_{~~s}$ and so the above equation becomes

\beq \hat W_{[3,s]}m_3^{~~3} = m_{3s,u} - m_{33,u} m_{3s} m_3^{~~3}. \nonumber \eeq

\noindent Dividing through by $m_3^{~~3} = m_{33}^{-1}$ leads to the desired identity

\beq \hat W_{[3,s]} = m_{3s,u}m_{33} - m_{33,u} m_{3s}. \nonumber \eeq

\noindent To show the next identity take $\Gamma_{3n3}$ , the upper-triangular form leads to the simpler expression

\beq  m_{3[r,3]} m_n^{~~r} m_3^{~~3} = 0. \nonumber \eeq

\noindent Since $m_n^{~~r}$ is invertible, \eqref{C2m3rcon1} follows from this identity. Finally, taking $\Gamma_{3np}$

\beq m_{3[e,f]}m_n^{~~e}m_p^{~~f} + m_{n r,3}m_p^{~~r}m_3^{~~3} + m_{p r,3}m_n^{~~r}m_3^{~~3} = 0. \nonumber  \eeq

\noindent From the above identity $m_{3[3,r]} = 0$ this simplifies to be

\beq m_{3[r,s]}m_n^{~~r}m_p^{~~s} = - (m_{n f,3}m_p^{~~f} + m_{p f,3}m_n^{~~f}) m_3^{~~3} \eeq

\noindent but $m_{n f,3}m_p^{~~f} + m_{p f,3}m_n^{~~f} = m_n^{~~r} m_p^{~~s} g_{rs,3}$. Substituting this into the left-hand side we have

\beq m_{3[r,s]} m_n^{~~r} m_p^{~~s} = - m_n^{~~r} m_p^{~~s} g_{rs,3} m_3^{~~3}. \nonumber \eeq

\noindent The matrix $m_{3[r,s]}$ is anti-symmetric and $g_{rs}$ is symmetric, by symmetrizing the above we find \eqref{C2m3rcon2} and \eqref{gcon1} hold.

\end{proof}

Equations \eqref{C2m3rcon1} and \eqref{C2m3rcon2} imply that $m_{3e} = M_{,e}$ for some $M(u,x^k)$ which will be at least a function of the spatial coordinate $x^3$, otherwise the component $m_{33}$ vanishes and the matrix $m_{ie}$ is no longer invertible. Interestingly, the matrix components $m_{nr}$ will be independent of $x^3$ due to \eqref{gcon1} and the fact that $m_{n r,3}m_p^{~~r} $ is upper triangular, since for $n<p$ we have

\beq m_n^{~~r} m_p^{~~s} g_{rs,3} = m_{n r,3}m_p^{\ r} = 0. \nonumber \eeq

\noindent In addition, the requirement that $g_{rs,3}= m_{nr,3} = 0$ give further constraints on $M$ ; by expanding the metric $g_{rs} = m_{ir} m^i_{~~s}$ and differentiating we see that

\beq g_{rs,3} = (M_{,r}M_{,s})_{,3} + (m_{nr}m^n_{~~s})_{,3} = 0. \eeq

\noindent Choosing $s=r$ this becomes $M_{,r3}M_{,r} = 0$ and so $M$ must be a function of $x^3$ and possibly coordinate $u$. The vanishing of $m_{3r} = M_{,r}$ along with \eqref{C2Wcon1} and \eqref{Aijeqn} imply that $A_{3n} = 0$. This equation will be particularly helpful in the subsequent cases as an equation for the $W_n = m_n^{e} \hat W_e$, which in expanded form is

\beq - D_{[3} W_{n]} + W^k D_{k3n} = 0. \label{case2wn}\eeq

\noindent However, by looking at the definition of $D_{k3n}$ we see that it vanishes.

Collecting the above results we have the following proposition

\begin{prop}\label{prop:metricfn1}

The vanishing of $\Gamma_{3ia}$ imply the upper-triangular matrix $m_{ie}$ arising from the transverse metric of \eqref{CCNVKundt} takes the form,

\beq & m_{33} = M_{,3}(u, x^3), ~~~ m_{3r} = 0, ~~~ m_{nr} = m_{nr}(u,x^r).& \label{mei} \eeq

\noindent While $W_n$ must satisfy

\beq   D_3 (W_n) = 0. \label{Wndetermined} \eeq

\end{prop}

The remaining Killing equations are then:

\beq &D_2X_2 + J_3 X_3 =0 & \label{case2ke1}\\
& D_3X_2 + D_2X_3 - J_3 X_1- B_{33}X_3 = 0 & \label{case2ke2} \\
& D_nX_2 - J_n X_1  = 0 & \label{case2ke3} \\
& D_3X_3 + B_{(33)}X_1 = 0& \label{case2ke4} \\
& D_nX_3 + 2B_{(3n)}X_1  = 0& \label{case2ke5} \\
& B_{(mn)}X_1 = 0.& \label{case2ke6}  \eeq

\noindent From equation \eqref{case2ke6} we have two subcases to consider. Either $X_1 = 0$ or $B_{(mn)} = 0$.

\end{subsection}

\begin{subsection}{Case 2.1: $X_1 = 0$, $B_{(mn)} \neq 0$}



\noindent If $F_3=0$, $F_2$ becomes a constant and so $X$ is some scaling of $\ell$, we will assume that $F_3 \neq 0$. From the third equation \eqref{case2ke3}

\beq  D_n F_2 = m_n^{~~e}F_{2,e} = 0. \nonumber \eeq

\noindent Multiplying by $m^n_{~~f}$ so that $m_n^{~~e}m^n_{~~f} = \delta^e_{~~f} - m_3^{~~e} m^3_{~~f}$

\beq F_{2,r} = m_3^{~~3}m^3_{~~r} F_{2,3}. \eeq

\noindent By Proposition \eqref{prop:metricfn1} we see that $m^3_{~~r} = 0$ and so the left hand side vanishes implying that $F_2$ is independent of all spacelike coordinates except possibly $x^3$.

\noindent Thus the remaining components of $X$ will be the following arbitrary functions,

\beq X_2 &=& F_2(u,x^3) \\
X_3 &=& F_3(u). \eeq

\noindent Expanding equation \eqref{case2ke2} we find that the constraining equation for $m_{33}$ is

\beq \frac{m_{33,u}}{m_{33}}  = \frac{- D_3 F_2 - D_2F_3}{F_3}. \label{case21m33} \eeq

\noindent While from \eqref{case2ke1} we have

\beq H = - \int \frac{ m_{33} D_2F_2}{F_3} dx^3 + A_5(u, x^r) \label{case21H}. \eeq


\noindent Thus, \eqref{case21H} and \eqref{case21m33} are equations for $H$ and $m_{33}$. The only constraint given for the $W_n$ comes from Proposition \eqref{prop:metricfn1}, i.e, they are all independent of $x^3$. This is just Case 1.2 with $X_1 = 0$ and the additional constraints in Proposition \eqref{prop:metricfn1}.

\end{subsection}

\begin{subsection}{Case 2.2: $B_{(mn)} = 0$, $X_1 \neq 0$}

Thanks to Proposition \eqref{prop:metricfn1}, we may repeat a similar calculation as Case 1.1 except with $B_{(np)}$ to show that for $n<p$ the vanishing of $B_{(np)}$ implies

\beq m_{n r , u} = 0. \label{case22mmr}\eeq

\noindent Furthermore, by proposition \eqref{prop:metricfn1}, the special form of $m_{ie}$ implies that $m_r^{~~3} = 0$, the only non-zero component of the tensor $B$ is $B_{33}$.

\noindent Since $v \in (-\infty, \infty)$ we may expand the Killing equations into orders of $v$, using \eqref{kvcomps} and the definition of the frame derivatives \eqref{framed}, to find a system of equations for $F_1$. 

\beq & D_2D_2F_1 + J_3 D_3F_1 = 0& \label{f1case22a} \\
& D_3 D_2 F_1 + D_2 D_3F_1 - B_{33}D_3  F_1  = 0& \label{f1case22b} \\
& D_n D_2 F_1 = 0& \label{f1case22c} \\
& D_3 D_3 F_1 = 0& \label{f1case22d} \\
& D_n D_3 F_1 = 0. & \label{f1case22e} \eeq

\noindent Along with another system of equations involving $F_2$ and $F_3$

\beq & H D_2F_1  + D_2 F_2 + J_3 F_3 = 0& \label{kecase22a} \\
& H D_3F_1  + D_{3} F_{2} + D_2 F_3 - J_3 F_1- B_{33}F_3 = 0 &  \label{kecase22b} \\
& D_n F_2 + W_n D_2 F_1 - J_n F_1 = 0& \label{kecase22c} \\
&  D_3 F_3 + B_{33} F_1 = 0& \label{kecase22d} \\
& W_n D_3 F_1 + D_n F_3 = 0.&  \label{kecase22e} \eeq

To begin, the special form of $m_{ie}$ from Proposition \eqref{prop:metricfn1} along with  equation \eqref{f1case22e} lead to the conclusion that $F_1$ must be independent of $x^3$ or $x^r$ . Note that if $F_{1,3} = 0$ then we have Case 1, with the added constraints $\Gamma_{3n2} = \Gamma_{3mi} = 0$. The analysis is not difficult, Case 1.23 may be omitted since $m_{3r}= 0$ while in Case 1.21 and 1.22, equation \eqref{mnr,u} is satisfied immediately, \eqref{m3r,u} now implies that $F_{3,r} = 0$. Only equation \eqref{m33,u} still holds, these cases are given in the table at the end of this section

It will be assumed that $F_{1,3} \neq 0$, then by expanding equation \eqref{f1case22d} the following relation between $F_1$ and $m_{33}$ is found

\beq \frac{m_{33,3}}{m_{33}} = \frac{F_{1,33}}{F_{1,3}}. \label{m33comma3} \eeq

\noindent Rewriting the term $D_2 D_3 (F_1) $ in \eqref{f1case22b} using the commutation relations \eqref{com} with $i=3$ as

\beq D_2 D_3 F_1 = D_3 D_2 F_1 - B_{33} D_3 F_1 \label{d2d4f1} \eeq

\noindent this is substituted into \eqref{f1case22b} yielding

\beq 2(D_3 D_2 F_1 - B_{33}D_3  F_1 ) = 0. \label{f1case22bnew} \eeq

\noindent The expanded form of \eqref{f1case22bnew} gives another relation between $m_{33}$ and $F_1$ 

\beq \frac{m_{33,u}}{m_{33}} = \frac{F_{1,3u}}{F_{1,3}}. \label{m33commau} \eeq

\noindent Thus $m_{33}(u,x^3)$ is entirely defined by $F_1$.

We may solve for $H$ and the $W_n$ algebraically from \eqref{kecase22b} and \eqref{kecase22e}

\beq H &=& \frac{D_3 D_2 F_1}{D_3(F_1)^2} F_3 - \frac{D_2^2 F_1}{D_3 (F_1)^2} F_1 - \frac{ 2D_{(2 } F_{ 3) } }{ D_3 F_1 } \label{case22H} \\
W_n &=& - \frac{D_n F_3}{D_3F_1}. \label{case22Wn} \eeq

\noindent Notice that by integrating \eqref{kecase22d}, $F_3$ is of the form:

\beq F_3 = \int \frac{m_{33} F_1 D_3D_2 F_1 }{D_3F_1} dx^3 + A_6(u, x^r) \label{f3case22} \eeq

\noindent Substituting \eqref{case22Wn} and \eqref{case22H} into \eqref{kecase22a} and \eqref{kecase22c} yields several equations for $F_2$

\beq & D_2 F_2 - \frac{D_2F_1 D_3 F_2 }{D_3F_1} =  D_2( \frac{F_3 D_2F_1 }{D_3F_1} ) - \frac{F_3 D_3(D_2F_1)^2 }{2(D_3F_1)^2} + \frac{F_1 D_2(D_2F_1)^2 }{2(D_3F_1)^2} & \label{f2case22a} \\
& D_3(F_1) D_3(F_1 D_nF_2 ) = D_3( F_1 D_2 F_1) D_nF_3. &  \label{f2case22b} \eeq

\noindent Hence $F_2(u,x^e)$ must depend on the choice of the arbitrary functions $F_1(u, x^3)$ and $A_6(u, x^r)$.

\end{subsection}

\begin{subsection}{Summary of Results} \label{SumOfResults}

We have considered the possibility of an additional Killing form in a $CCNV$ spacetime, where the metric functions $H$, $\hat W_i$ and $g_{ef}$. Given the arbitrary form of the $CCNV$ metric in equation \eqref{CCNVKundt}, we used a coordinate transformation \eqref{now3} to eliminate $\hat W_3$; this is done to simplify the constraints on the metric functions. Next the null frame \eqref{nullframe} was rotated so that $m^3$ is parallel with the spatial part of $X$. Due to the $QR$ decomposition it is always possible to treat the matrix, $m_{ie}$ as an upper-triangular matrix, this is assumed through-out the paper.

The first four equations \eqref{killeqn0} -- \eqref{killeqn3} imply that the components of the Killing co-vector are given by \eqref{kvcomps}. By applying the commutator relations for the frame derivatives \eqref{framed} to the Killing equation splits the analysis into two simpler cases, depending on whether $D_3 X_1=0$ or $\Gamma_{3m2}= \Gamma_{3mj} = 0$. While the first case requires that $X_1$ is independent of $x^3$, the implications of $\Gamma_{3m2} = \Gamma_{3mj} = 0$ lead to the  constraints on the $W_n$ and the matrix $m_{ie}$ given in Proposition \eqref{prop:metricfn1}.  Both cases are summarized in the tables below.

\begin{table}[ht]
\begin{center} 
\begin{tabular}{c|c c c c c c }
\hline
Case & $X_1$ & $F_2$ & $F_3$ & $m_{ie}$ & $H$ & $W_n$  \\ [0.5ex]
\hline
& & & & & & \\
1.11 & $u$ & \eqref{case11f2} & 0 & $m_{ie,u} = 0$ & \eqref{C11H} & \eqref{C11Wn} \\ [1ex]
1.12 & $1$ & $F_{2,u} = 0 $ & 0 & $m_{ie,u} = 0$  & \eqref{C11H0} & \eqref{C11Wn0} \\ [1ex]
1.21 &$u$ & $F_2$ & $F_3$ & \eqref{m33,u} -- \eqref{m3r,u} & \eqref{C12H} &  \eqref{C12Wn} \\ [1ex]
1.22 & $1$ & \eqref{C12f20} & $F_3$  & \eqref{m33,u} -- \eqref{m3r,u} & \eqref{C12H0} &  \eqref{C12Wn0} \\ [1ex]
1.23 & 0 & $F_2$ & $F_3$ & \eqref{mnr,u0}, \eqref{m33,u0} & \eqref{C12H00} & \eqref{C12Wn00} \\
$m_{3r,3} \neq 0$ & & & & \eqref{m3r,3} &  &  \\ [1ex]
1.24 & 0 & $F_2$ & $F_3$ & \eqref{mnr,u0}, \eqref{m33c121} & \eqref{C12H00} & \eqref{C12Wn00} \\
$m_{3r,3} =0$ & & & &  &  &  \\
\hline
\end{tabular}
\caption{ Summary of analysis in Case 1 }
\end{center}
\label{table:case1}
\end{table}

\begin{table}[ht]
\begin{center} 
\begin{tabular}{c|c c c c c c }
\hline
Case & $X_1$ & $F_2$ & $F_3$ & $m_{ie}$ & $H$ & $W_n$  \\ [0.5ex]
\hline
& & & & & & \\
2.1  &0 & $F_{2,r} = 0$ & $F_{3,e} = 0$ & \eqref{m33,u0} & \eqref{C12H00} & \eqref{C12Wn00} \\ [1ex]
2.21 & $u$ & \eqref{case11f2} & 0 & $m_{ie,u} = 0$ & \eqref{C11H} & \eqref{C11Wn} \\ [1ex]
2.22 & $1$ & $F_{2,u} = 0$ & 0 & $m_{ie,u} = 0$  & \eqref{C11H0} & \eqref{C11Wn0} \\[1ex]
2.23 & $u$ & $F_2$ & $F_{3,r} = 0 $ & \eqref{m33,u}, \eqref{case22mmr} & \eqref{C12H} &  \eqref{C12Wn} \\[1ex]
2.24 & $1$ & \eqref{C12f20} & $F_{3,r} = 0 $  & \eqref{m33,u},  \eqref{case22mmr} & \eqref{C12H0} &  \eqref{C12Wn0} \\[1ex]
2.25 & 0 & $F_2$ & $F_{3,e} = 0$ & \eqref{m33,u0}, \eqref{case22mmr} & \eqref{C12H00} & \eqref{C12Wn00} \\[1ex]
2.26 & $F_{1,r} = 0$ & \eqref{f2case22a} & \eqref{f3case22} & \eqref{m33commau}, \eqref{m33comma3}, & \eqref{case22H} & \eqref{case22Wn} \\
& & \eqref{f2case22b} & & \eqref{case22mmr} & & \\ [1ex]
\hline
\end{tabular}
\caption{Summary of Case 2, where Proposition \eqref{prop:metricfn1} implies $D_3(W_n) = 0$ and $m_{ie}$ takes the special form \eqref{mei} }
\end{center}
\label{table:case2}
\end{table}

\end{subsection}

\newpage

\begin{subsection}{Killing Lie Algebra}

In \eqref{SumOfResults} we found that there are only three particular forms for the Killing vector in those $CCNV$ spacetimes admitting an additional Killing vector, depending on the choice of $X_1$. The three cases depend on whether $X_1$ is linear in $u$, $X_1$ is a constant or $X_1$ is a function of $u$ and $x^3$. The remaining functions involved with $X_2$ and $X_3$ are functions of $u$ and $x^e$, satisfying the appropriate equations in the above two tables . We will label those spacetimes admitting an additional Killing vector by its type; using \eqref{kvcomps} we may write the three possible types for the Killing vector $X$ as

\beq  &X_A& = c n + F_2(u,x^e) \ell + F_3(u,x^e) m^3  \nonumber \\
 &X_B& = u n + [F_2(u,x^e)-v] \ell + F_3(u,x^e) m^3 \nonumber \\
 &X_C& = F_1(u,x^3) n + [F_2(u,x^e - D_2F_1 v] \ell + [F_3 - D_3F_1 v] m^3 . \nonumber \eeq

\noindent To see if these spacetimes admit even more Killing vectors we will examine each case and consider the commutator with $\ell$. Using the frame formalism, the commutator of two vector-fields $X = X^a e_a$, $Y = Y^b e_b$ is

\beq [X,Y] = X^a e_a(Y^b) - Y^a e_a(X^b) + 2X^a Y^c\Gamma^b_{[ac]} \label{FrameComm} \eeq

\noindent When $Y = \ell$ we have that $Y^c = \delta_1^c$ and from \eqref{com} $\Gamma^b_{1a} = 0$ so the commutator is

\beq [X, \ell ] = - \ell(X^b). \eeq

\noindent Thus in the Type A spacetimes there are no other required Killing vectors except $\ell$ and $X$. Similarly in the Type B spacetimes, the commutator of $\ell$ and $X$ is

\beq [X_B, \ell] = - \ell \eeq

\noindent this is just a scaling of a known vector so we may conclude in general that Type B Spacetimes contain no additional Killing vectors other than $\ell$ and $X$.

The most general case is more interesting because the commutator of $X$ and $\ell$ yields a new Killing vector

\beq Y_C = [\ell, X_C ] = D_2F_1 \ell + D_3F_1 m_3. \label{OhNoY} \eeq

\noindent  Clearly this will be a space-like Killing vector for all choices of $F$ since its magnitude is $|Y_C|$ = $(D_3F_1)^2 > 0$. The commutator of $Y_C$ with $\ell$ vanishes because $F_1$ is a function of $u$ and $x^3$, however the commutator of $Y_C$ and $X_C$ cannot in general be set to zero. A quick calculation gives $Z_C = [X_C,Y_C]$

\beq Z_C &=& [F_3 D_3D_2F_1 - D_3F_1 D_3F_2 + (D_2F_1)^2] \ell \nonumber \\
& & - (D_3F_1)^2 n + [D_2F_1 D_3F_1] m_3 \label{OhNoZ} \eeq

Thus because we assumed $F_1 \neq 0$ and $D_3F_1 \neq 0$ we may never set $Z_C = 0$ due to the coefficients of $n$. The Type C spacetimes admit at least one additional spacelike Killing vector.

\end{subsection}

\end{section}

\begin{section}{$CSI$ $CCNV$ spacetimes possessing an additional Killing vector} \label{CSICCNV}

In section \ref{CSICriteria} it was shown that if a $CCNV$ spacetime has constant scalar curvature invariants to all orders its transverse metric $g_{ef}$ must be locally homogeneous. Applying this result we may write down the constraints for a $CSI$ $CCNV$ spacetime to admit an additional Killing vector by choosing an appropriate locally homogeneous Riemannian manifold for the transverse space $g_{ef}$. This choice may affect the components of the Killing vector $X$.

Due to the local homogeneity of $g_{ef}$ one may perform a coordinate transformation so that the matrix $m_{ie}$ is independent of $u$. Looking at the tables in section \ref{SumOfResults} we note that Cases 1.11 - 1.13 and 2.21 - 2.23 already require that $m_{ie}$ be independent of $u$ and there are no constraints on the Killing vector components involving $m_{ie}$. Therefore the $CSI$ spacetimes are the subcases of these cases where the transverse metric is a locally homogeneous. The remaining cases are more interesting since they involve a non-zero spatial component of $X$.

\begin{subsection}{Case 1} \label{CSICase1}

In Case 1.2 equation \eqref{m33,uode} implies that

\beq F_{3,3} = 0 \label{f33bye}\eeq

\noindent while from \eqref{mnr,u}

\beq m_{nr,3} = 0. \label{gef0} \eeq

\noindent The remaining equations \eqref{mnr,u}, \eqref{m33,uode}, \eqref{m3r,u} arose from \eqref{3b2case1}, this may be rewritten as a differential equation for $F_3$

\beq D_n(log F_3) = \Gamma_{3n3}.  \label{ohgoodlord} \eeq

\noindent It is possible to derive even more constraints on the transverse space through the commutation relations $[D_i, D_n]$ applied to $F_3(u,x^r)$. To start, we note that because of \eqref{com} and the $u$-independence of $m_{ie}$, $[D_2, D_n](logF_3) $ becomes

\beq D_nD_2(logF_3) = 0. \label{ohgoodson} \eeq

\noindent Next we consider $[D_3, D_n]$, since  $D_3 (log F_3) = 0$ and $D_3(m_n^{~r})= 0 $ the commutator is

\beq [D_3, D_n](logF_3) = D_3(\Gamma_{3n3}) = 0. \nonumber \eeq

\noindent However using \eqref{com} with $k=3$ we find that

\beq [D_3, D_n](logF_3) = \Gamma^m_{~[3n]} D_m(logF_3). \nonumber \eeq

\noindent Thus we find two new  constraints on the connection coefficients arising from the transverse metric

\beq D_3(\Gamma_{3n3}) = 0 \label{gef1} \\
\Gamma^m_{~[3n]} \Gamma_{3m3} = 0 \label{gef2}. \eeq

\noindent With $k, j > 3$ in \eqref{com} we find that

\beq [D_n, D_m](logF_3) = 2\Gamma^p_{~[nm]}\Gamma_{3p3} \nonumber \eeq

\noindent whereas from \eqref{ohgoodlord} we have that

\beq [D_n, D_m](logF_3) = D_{n}\Gamma_{3m3} - D_{m}\Gamma_{3n3}. \nonumber \eeq

\noindent Equating the two gives another constraint on the transverse space,

\beq \Gamma^p_{~[nm]}\Gamma_{3p3} =  D_{n}\Gamma_{3m3} - D_{m}\Gamma_{3n3}. \label{gef3} \eeq

\noindent Thus if we require the Killing vector to have a spatial component the connection coefficients arising from the transverse metric must satisfy the equations \eqref{gef0}, \eqref{gef1}, \eqref{gef2} and \eqref{gef3}. We will assume such a transverse metric has been found in order to continue with the analysis.

Reconsidering  \eqref{m3r,u} and solving for $m_{3r,3}$ leads to another differential equation

\beq m_{3r,3} = F_3 \left( \frac{m_{33} }{F_3} \right )_{,r} \label{1CSIsplitter} \eeq

and two possibilities,  either $m_{3r,3}$ vanishes or not.

If $m_{3r,3} \neq 0$ we may integrate the above equation to find an expression for $m_{3r}$

\beq m_{3r} = \int \left( \frac{m_{33}}{F_3} \right)_{,r} F_3 dx^3 + B_r(x^s). \label{1CSIm3ra} \eeq

\noindent Differentiating with respect to $u$ we must have

\beq  \left( \frac{ F_{3,r} }{F_3} \right)_{,u}  = 0. \label{1CSIm33a} \eeq

\noindent So that $F_3$ is of the form

\beq log(F_3) = g_3(u) + f_3(x^r). \label{1CSIf3} \eeq

\noindent The form of $log(F_3)$ above agrees with the differential equation given in \eqref{ohgoodlord} and \eqref{ohgoodson}, where we expect $f_3$ is determined by the $\Gamma_{3n3}$. In this case, \eqref{gef0} and \eqref{1CSIm3ra} are the only equations for the transverse metric so far. The remaining constraints on the metric functions will vary for each subcase depending on the choice of $X_1 = c_1 u + c_2$.

If $c_1 \neq 0$, we see that $H$ is given by

\beq H = - \frac{1}{u} (u D_2F_2 + F_3 D_3F_2 + F_3 D_2F_3 ). \label{1CSIH0} \eeq

\noindent While the $W_n$ satisfy the determining equation

\begin{equation}
D_{2}( u W_n ) + F_3 D_3 W_n + D_n(F_2 - u H ) = 0 \, . \label{1CSIWn0}
\end{equation}

If $c_1 = 0$, $F_2$ is no longer arbitrary it must satisfy the following equation

\beq D_2F_2 + F_3 D_3F_2 + F_3 D_2F_3   = 0 \label{1CSIf2} \eeq

\noindent $H$ may be written as


\beq H &=& F_2 + \int m_{33}D_2F_3 dx^3 +  A_1(u, x^r), \label{1CSIH1} \eeq

\noindent and the equation for $W_n$ is now

\beq &D_2 W_n + F_3 D_3W_n  = D_nF_2 - D_nH . \label{1CSIWn1} & \eeq

If $c_1 = c_2 = 0$ then $X_1$ vanishes, we know that this turns \eqref{m33,u} -- \eqref{m3r,u} into the same set of differential equations \eqref{f33bye}, \eqref{gef0} and \eqref{ohgoodlord}.

The remaining metric functions $H$ and the $W_n$ are

\beq H = - \int \frac{ m_{33} D_2F_2}{F_3} dx^3 + A_2(u, x^r), \label{1CSIH2} \\
W_n = - \int \frac{ m_{33} D_nF_2}{F_3} dx^3 + E_n(u, x^r). \label{1CSIWn2} \eeq

If $m_{3r,3} = 0$ we find the following expression for $m_{33}$

\beq D_n \left( \frac{ m_{33} }{ F_3 } \right) = 0. \eeq.

\noindent This leads to the following equality

\beq D_n(log m_{33}) = D_n(log F_3) = \Gamma_{3n3} \eeq

\noindent Expanding this out we find the constraint,

\beq m_{33,3}m^3_{~r} = (m^2_{33})_{,r}  \label{gef4} \eeq

The equations for the metric functions follow as above in the various cases arising from $X_1 = c_1 u + c_2$.

\end{subsection}

\begin{subsection}{Case 2}

If $m_{ie,u}$ vanishes, Case 2.1 and Case 2.26 are now the same case. Due to Proposition \ref{prop:metricfn1}, equations  \eqref{mnr,u} \eqref{m33,uode} \eqref{m3r,u} imply that

\beq F_{3,e} = 0. \eeq

>From \eqref{m33,u0} we find the familiar equation for $F_2$:

\beq F_{2,3} = - m_{33} F_{3,u}. \label{2CSIm33} \eeq

\noindent The metric function $H$ is given by equation \eqref{case21H} and $W_n$ may be arbitrary functions of $u$ and $x^r$.

In Cases 2.24 - 2.26 the vanishing of $m_{ie,u}$ causes \eqref{m33,u} and \eqref{m3r,u} to imply

\beq F_{3,e} = 0. \label{2CSIf3a} \eeq

\noindent So $F_3$ is only a function of $u$. With this in mind the equations for the remaining metric functions are the same as in \eqref{CSICase1}, $m_{3r,3} = 0$, with the additional constraints from Proposition \eqref{prop:metricfn1}.

In Case 2.27, equation \eqref{m33commau} now implies

\beq F_{1,u3} = 0. \eeq

\noindent The function $F_1$ must be of the form

\beq F_1 = g_1(u) + f_1(x^3) \label{C2f1a} \eeq

\noindent $g_1$ is an arbitrary function of $u$ however, $f_1$ satisfies the differential equation in \eqref{m33comma3}

\beq D_3(log m_{33}) = D_3 log( f_{1,3}). \nonumber \eeq

\noindent Letting $c_3$ be arbitrary constant we have by integrating that

\beq f_1 =  m_{33} + c_3  \label{C2f1b} \eeq

\noindent Combining \eqref{C2f1a} and \eqref{C2f1b} yields

\beq F_1 = m_{33} + g_1 + c_3 \label{2CSIf1} \eeq

>From \eqref{f3case22}, $F_3$ must be a function of $u$ and $x^r$, the equations \eqref{f2case22a} and \eqref{f2case22b} are

\beq &D_2 F_2 - \frac{D_2F_1}{D_3F_1} D_3 F_2 = (\frac{D^2_2F_1}{D_3F_1})F_3 + \frac{D_2F_1}{D_3F_1} D_2F_3 + \frac{F_1 D_2F_1 D^2_2F_1 }{D_3F_1^2},& \label{2CSIf2a} \\
&D_3(F_1 D_n F_2) = D_2F_1D_nF_3.& \label{2CSIf2b} \eeq

\noindent Assuming $g_1' \neq 0$, dividing through \eqref{2CSIf2a} by $\frac{ g_1' }{\dot f_1}$ it is possible to solve for $D_2F_3$, while dividing from $g_1'$ in \eqref{2CSIf2b} we have $D_nF_3$, hence it is possible to solve for $F_3$ entirely in terms of $F_2$ and $g_1'$.

Simplifying \eqref{case22H} and \eqref{case22Wn} we find that,

\beq H &=& - \frac{D_2^2 F_1}{D_3 (F_1)^2} F_1 - \frac{D_{2} F_3}{D_3 F_1} - \frac{D_{3} F_2}{D_3 F_1} \label{2CSIH} \\
W_n &=& - \frac{D_n F_3}{D_3F_1}. \label{2CSIWn} \eeq

If $D_2F_1 \neq 0$, dividing \eqref{2CSIf2a} by this then substituting this into the equation for $H$ gives,

\beq H &=& (\frac{D^2_2 F_1}{D_3F_1})F_3 - D_2F_2. \eeq

\noindent We note that $F_2$ may be entirely an arbitrary function of $u$ and $x^3$.

\end{subsection}

\begin{subsection}{Summary of Constraints}

As in the previous section we summarize our results for the existence of an additional null Killing covector in a $CSI$ $CCNV$ spacetime. In order for a $CCNV$ spacetime to be $CSI$, we found in section \ref{CSICriteria} that the transverse space must be locally homogeneous. This allows one to choose a coordinate chart locally such that

\beq m_{ie,u} = 0 \nonumber \eeq

\noindent so that many of the differential equations given in the previous section are simpler. The results of this analysis are summarized below in the two tables. In all cases the transverse metric is locally homogeneous, although it should be noted that in subsection \eqref{CSICase1}, if $F_3 \neq 0$ then the transverse space must satisfy the following constraints

\beq m_{nr,3} &=& 0 \nonumber \\
 D_3(\Gamma_{3n3}) &=& 0 \nonumber \\
\Gamma^m_{~[3n]} \Gamma_{3m3} &=& 0 \nonumber \\
\Gamma^p_{~[nm]}\Gamma_{3p3} &=&  D_{n}\Gamma_{3m3} - D_{m}\Gamma_{3n3}. \eeq

\noindent We also remind the reader that if $m_{3r,3} = 0$ equation \eqref{gef4} holds, so that $m_{33}$ is independent of $x^r$. In the second case the matrix $m_{ie}$, related to the locally homogeneous transverse metric, must satisfy \eqref{mei} in Proposition \eqref{prop:metricfn1}.

\newpage

\begin{table}[ht]
\begin{center} 
\begin{tabular}{c|c c c c c }
\hline
Case & $X_1$ & $F_2$ & $F_3$ & $H$ & $W_n$  \\ [0.5ex]
\hline
& & & & & \\
1.11 & $u$ & \eqref{case11f2} & 0  & \eqref{C11H} & \eqref{C11Wn} \\ [1ex]
1.12 & $1$ & $F_{2,u} = 0 $ & 0 & \eqref{C11H0} & \eqref{C11Wn0} \\
\hline
\end{tabular}
\caption{ Summary of Killing equations analysis in Case 1 for a $CSI$ $CCNV$ spacetime, when $F_3 = 0$}
\end{center}
\label{table:CSIcase1a}
\end{table}

\begin{table}[ht]
\begin{center} 
\begin{tabular}{c|c c c c c }
\hline
Case & $X_1$ & $F_2$ & $F_3$ & $H$ & $W_n$  \\ [0.5ex]
\hline
& & & & & \\
1.21a &$u$ & $F_2$ & \eqref{1CSIf3} &  \eqref{1CSIH0} &  \eqref{1CSIWn0} \\ [1ex]
1.22a & $1$ & \eqref{1CSIf2} & \eqref{1CSIf3} & \eqref{1CSIH1} &  \eqref{1CSIWn1} \\ [1ex]
1.23a & 0 & $F_2$ & $F_3$ & \eqref{1CSIH2} & \eqref{1CSIWn2} \\
$m_{3r,3} \neq 0$  & & & \eqref{1CSIm3ra} &  &  \\
\hline
\end{tabular}
\caption{ Summary of Killing equations analysis in Case 1 for a $CSI$ $CCNV$ spacetime, when $F_3 \neq 0$ and $m_{3r,3} \neq 0$}
\end{center}
\label{table:CSIcase1b}
\end{table}

\begin{table}[ht]
\begin{center} 
\begin{tabular}{c|c c c c c }
\hline
Case & $X_1$ & $F_2$ & $F_3$ & $H$ & $W_n$  \\ [0.5ex]
\hline
& & & & & \\
1.21b &$u$ & $F_2$ & $F_{3,3}=0$ & \eqref{1CSIH0} &  \eqref{1CSIWn0} \\ [1ex]
1.22b & $1$ & \eqref{1CSIf2} & $F_{3,3}=0$  & \eqref{1CSIH1} &  \eqref{1CSIWn1} \\ [1ex]
1.23b & 0 & $F_2$ & $F_{3,3} = 0$ & \eqref{1CSIH2} & \eqref{1CSIWn2} \\
\hline
\end{tabular}
\caption{ Summary of Killing equations analysis in Case 1 for a $CSI$ $CCNV$ spacetime, when $F_3 \neq 0$ and $m_{3r,3} = 0$}
\end{center}
\label{table:CSIcase1c}
\end{table}

\begin{table}[ht]
\begin{center} 
\begin{tabular}{c|c c c c c c }
\hline
Case & $X_1$ & $F_2$ & $F_3$ & $H$ & $W_n$  \\ [0.5ex]
\hline
& & & & & & \\
2.1 &0 & $F_{2,r} = 0$ & $F_{3,e} = 0$ & \eqref{C12H00} & \eqref{C12Wn00} \\
 & & & & $m_{ie,u} = 0$ & & \\[1ex]
2.21& $u$ & \eqref{case11f2} & 0 & \eqref{C11H} & \eqref{C11Wn} \\
$F_{2,3} \neq 0$ & & & & & &  \\ [1ex]
2.22& $1$ & $F_{2,u} = 0$ & 0 & \eqref{C11H0} & \eqref{C11Wn0} \\[1ex]
2.23& $u$ & $F_2$ & $F_{3,e} = 0 $ &  \eqref{1CSIH0} &  \eqref{1CSIWn0} \\ [1ex]
2.24& $1$ & \eqref{1CSIf2} & $F_{3,e} = 0 $  & \eqref{1CSIH1} &  \eqref{1CSIWn1} \\ [1ex]
2.25& \eqref{2CSIf1} & \eqref{2CSIf2a} & $F_{3,3} = 0$ & \eqref{2CSIH} & \eqref{2CSIWn} \\
& & \eqref{2CSIf2b} & & & & \\ [1ex]
\hline
\end{tabular}
\caption{Summary of Case 2 for a $CSI$ $CCNV$ spacetime - Proposition \eqref{prop:metricfn1} implies $D_3(W_n) = 0$ and $m_{ei}$ takes the special form \eqref{mei} }
\end{center}
\label{table:CSIcase2}
\end{table}

\end{subsection}

\end{section}

\begin{section}{Application: Non-spacelike isometries}

The metric for $CCNV$ spacetimes \eqref{CCNVKundt} must be independent of $v$, varying this coordinate value leaves the metric unchanged. With regards to the set of $CCNV$ spacetimes admitting an additional Killing vector, a good question to ask is: which of these spacetimes admit a non-spacelike Killing vector for all values $v$? In \cite{mcnutt} this  was considered for $CSI$ $CCNV$ spacetimes, however the approach taken differs from the one presented in this paper.

While the frame was rotated so that the Killing covector $X$ has one spatial component $X_3$ and the matrix $m_{ie}$ is upper-triangular, a coordinate transformation was made to eliminate $H$ instead of $W_3$. Regardless of these coordinate changes, the equations \eqref{killeqn0} -- \eqref{killeqn3} lead to the same form for the Killing covector components given in \eqref{kvcomps}. The non-spacelike requirement for the Killing vector field maybe written as

\beq D_3(X_1)^2 v^2 + 2(D_2(X_1) X_1 - D_3(X_1) F_3) v + F_3^{~2} - 2 X_1 F_2 \leq 0 \nonumber \eeq

\noindent Since $v \in (-\infty, \infty)$ this implies that $D_3(X_1)$ must vanish and either $X_1$ is independent of $u$ or $X_1 = 0$.
Thus either $X_1$ is constant or it vanishes entirely. This requirement along with the $u$ independence of $m_{ie}$ from the $CSI$ condition lead to a simpler set of equations for the remaining components of $X$, as such the commutator relations were ignored in \cite{mcnutt} and the analysis was done using the coordinate basis instead of the frame formalism. As such the results of \cite{mcnutt} agree with the results given in this paper, but only as special subcases of 1.1 and 1.2 in the Table \ref{table:nullkv} given below.

Instead of the labour intensive approach given in \cite{mcnutt} we may use the result of the previous section to find an answer to the question of non-spacelike Killing vectors in $CCNV$ spacetimes. Since $X_1$ must be constant from the non-spacelike requirement we have that Cases 1.11, 1.21, 2.21, 2.23 and 2.26 are no longer admissible. In the remaining cases the only constraint left is for $F_2$ and $F_3$ is

\beq F_3^{~2} - 2 X_1 F_2 \leq 0. \label{ntl} \eeq

\noindent Hence we will divide the analysis into two cases depending on whether the vector is timelike or null.

\begin{subsection}{Timelike Killing vector fields}

If we allow $X$ to be a timelike Killing vector field, we have the constraint that

\beq F_3^2 < 2c_2 F_2 \label{F3LessThanF2} \eeq

\noindent and so the cases with $X_1=0$ (1.23,1.24,2.1 and 2.25) are no longer valid since $F_3$ is a real-valued function and with $F_3^2 < 0.$ which is impossible and so these cases will be disregarded. In the remaining cases (1.12, 1.22, 2.22 and 2.24) equation \eqref{F3LessThanF2} is an additional constraint on $F_3$ and $F_2$. Thus the Killing vector field $X$ will always be of the form

\beq n + F_2(u,x^e) \ell + F_3(u,x^r) m^3. \eeq

\noindent The requirement that $F_3^2 < 2F_2$ does not affect the equations in the various cases.

\end{subsection}

\begin{subsection}{Null Killing vector fields}

If $X$ is null, and $c_2 \neq 0$ we can rescale $n$ so that \eqref{ntl} implies that $2F_2 = F_3^{~2}$, from which we naturally find the helpful identity

\beq D_a(F_2) = D_a(F_3) F_3. \label{idnull} \eeq

\noindent If $F_3$ vanishes as in Case 1.12, $F_2$ must vanish as well, so $X$ takes the form

\beq X = n. \label{nullX0} \eeq

\noindent the remaining equations for the metric functions are now

\beq H &=& A_0(u,x^r) \label{nullH0} \\
W_n &=& \int D_n(A_0) du + C_n(x^e) \label{nullWn0}. \eeq

\noindent The transverse metric is unaffected by \eqref{ntl}.

In Case 1.22, $X$ is now

\beq X = n + \frac{F_3^2}{2} \ell + F_3 m^3 \eeq

\noindent taking equation \eqref{C12f20} we find a differential equation for $F_3$

\beq D_2(F_3) + D_3(F_3) F_3 = 0. \label{nullodef3} \eeq

\noindent This allows us to rewrite $H$ as

\beq H &=& A_2(u,x^r) \label{nullH1}. \eeq

\noindent The constraining equation for the $W_n$ is

\beq D_2(W_n) + D_3(W_n)F_3 = D_n(A_2) \label{nullWn1}. \eeq

\noindent We may rewrite \eqref{m33,uode} in a simpler form using \eqref{nullodef3}

\beq \frac{m_{33,u}}{m_{33}}= \frac{ D_2(F_3)}{F_3}.  \label{nullm33} \eeq

\noindent This may be integrated to find $m_{33}$ in terms of $F_3$, assuming $m_{33,u} = 0$, however this is more useful as a differential equation.

\noindent The remaining two equations \eqref{mnr,u} and \eqref{m3r,u} are unchanged.

If $c_2 = 0$, then $X_1$ vanishes entirely and \eqref{ntl} implies that

\beq F_3^2 = 0. \eeq

\noindent If $X_1 = F_3 = 0$ then the Killing equations \eqref{1bcase1} -- \eqref{3b3case1} implies $F_2$ must be constant. That is, our Killing covector is a scalar multiple of $\ell$, so we will disregard this as well as Case 2.1 and 2.25.

The remaining cases 2.22 and 2.24 are just a repetition of the above equations with the added constraints that Proposition \eqref{prop:metricfn1} holds and $m_{mn,u} = 0$. In the first case where $F_3 = 0$, this changes the $W_n$

\beq W_n =  \int D_n(A_0) du + C_n(x^r) \label{nullWn2}. \eeq

\noindent No other metric functions are affected. When $F_3 \neq 0$, the additional constraints imply that \eqref{mnr,u} is satisfied trivially and \eqref{m3r,u} becomes

\beq F_{3,r} = 0. \label{nullm3r} \eeq

\noindent Lastly since $D_3(W_n) = 0$, equation \eqref{nullWn1} implies

\beq W_n  = - \int D_n(A_2) du + A_7(x^r). \label{nullWn3} \eeq

\noindent We summarize these results in the following table.

\begin{table}[ht]
\begin{center} 
\begin{tabular}{c|c c c c c c }
\hline
Case & $X_1$ & $F_2$ & $F_3$ & $m_{ie}$ & $H$ & $W_n$  \\ [0.5ex]
\hline
& & & & & & \\
1.1 & 1 & 0  & 0 & $m_{ie,u} = 0$  & \eqref{nullH0} & \eqref{nullWn0} \\ [1ex]
1.2 & 1 & $\frac{1}{2}F_3^2$ & $F_3$  & \eqref{nullm33}, \eqref{m3r,u} & \eqref{nullH1} &  \eqref{nullWn1} \\
 & & & & \eqref{mnr,u} & & \\ [1ex]
2.1 & 1 & 0  & 0 & $m_{ie,u} = 0$  & \eqref{nullH0} & \eqref{nullWn2} \\
 & & & & \eqref{mei} & & \\ [1ex]
2.2 & 1 & $\frac{1}{2}F_3^2$ & $F_{3,r} = 0$  & \eqref{nullm33}, \eqref{mei} & \eqref{nullH1} &  \eqref{nullWn3} \\
 & & & & $m_{nr,u}=0$ & & \\ [1ex]
\hline
\end{tabular}
\caption{ Constraints on $CCNV$ metric in order to allow null isometry. }
\end{center}
\label{table:nullkv}
\end{table}

If we wish to find $CSI$ $CCNV$ spacetimes admitting Killing vectors which are non-spacelike for all values of $v$, the above table will be helpful. The $CSI$ $CCNV$ spacetimes are the subcases of the above cases, where the transverse space is locally homogenous, allowing for a choice of coordinates where $m_{ie}$ is independent of $u$.

In case 1.1 and 2.1, none of the equations are affected by the vanishing of $m_{ie,u}$, while in case 1.2 and 2.2 equation \eqref{nullm33} is no longer applicable, instead we look to \eqref{m33,u} which implies that

\beq F_{3,3} = 0. \eeq

Unlike the previous cases,  where no other function is affected by our choice of coordinates, in case 1.2 \eqref{m3r,u} implies

\beq m_{mr,3} = 0. \eeq

\noindent From \eqref{mnr,u} we find the following differential equation for $F_3$

\beq D_n(log F_3) = \Gamma_{3n3}. \eeq

\noindent The remaining constraints on the metric functions follows as in section \ref{CSICCNV} where the Killing vector $X$ is of type A with some minor modifications due to $F_2 = \frac12 F_3^2$. We will do a simple example to illustrate.
\end{subsection}

\begin{subsection}{A simple example}

To simplify matters, we will assume that the transverse space is locally homogeneous and that $\Gamma_{3n3} $, $m_{3r,3}$ both vanish.  By Lemma \ref{lem:ccnvcsi} this will be a $CSI$ $CCNV$ spacetime and since \eqref{gef0}, \eqref{gef1}, \eqref{gef2} and \eqref{gef3} are all satisfied, it will also admit  a null Killing vector $X$ of the form

\beq X = n + \frac{ F^2_3 }{ 2 } \ell + F_3 m^3 \eeq

\noindent For brevity we will only consider the simpler subcase where it is assumed that the components $m_{3r}$ are independent of $x^3$. Expanding $\Gamma_{3n3}$ in terms of the transverse space frame matrix

\beq \Gamma_{3n3} = m_{3[r,3]}m_n^{~r}m_3^{~3} = -m_{33,r}m_n^{~r}m_3^{~3} \eeq

\noindent Then by multiplying this by $m^n_{~r}$ we find that

\beq m_{33,r} = 0. \eeq

\noindent This and equation \eqref{gef4} then implies that

\beq m_{33,3} m^3_{~r} = 0. \eeq

\noindent So that either $m_{33,3} $ or $m^3_{~r}$ must vanish. If $m^3_{~r} = 0$, the matrix $m_{ie}$ will take the form

\beq m_{33} = m_{33}(x^3) \nonumber \\
m_{3r} = 0 \label{mienullcsi1}\\
m_{nr} = m_{nr}(x^r). \nonumber \eeq

\noindent On the other hand, if $m_{33,3} = 0$, it will be a constant, say $M_{33}$, and then the matrix $m_{ie}$ is of the form

\beq m_{33} = M_{33} \nonumber \\
m_{3r} = m_{3r}(x^r) \label{mienullcsi2}\\
m_{nr} = m_{nr}(x^r). \nonumber \eeq

\noindent In either case, the choice does not affect the remaining metric functions and Killing vector components.

Noting \eqref{f33bye} in \eqref{ohgoodlord},  we may multiply by $m^n_{~s}$ to see that $F_3$ is at most a function of $u$. However, from equation \eqref{nullm33} we see that it must be a constant. By requiring $X$ to be null, we obtain  $F_2 = \frac{ F^2_3 }{ 2 } $, and so \eqref {1CSIf2} gives no new information. The Killing vector may then be written as

\beq n + \frac12 \ell + m^3, \eeq

\noindent subtracting the known Killing vector $\frac12 \ell$ we find the spacelike Killing vector, $Y = n + m^3$.
The metric function $H$ is found to be an arbitrary function of $u$ and $x^r$ by  \eqref{nullH1}, while $W_n$ is determined by the linear partial differential equation in \eqref{nullWn1}

\beq D_3W_n + D_2 W_n = D_nH. \eeq

\noindent Rewritting the above in coordinate form

\beq \hat W_{r,3} + \hat W_{r,u} = H_{,r}. \eeq

\noindent Applying the method of characteristics, the solution is written as

\beq \hat W_r = \frac{ 1 }{ \sqrt{2} } \int_L H(u, x^r) ds + g(x^3-u) \eeq

\noindent where $g$ is an arbitrary function of one variable and $L$ is the characteristic line segment from the $u$-axis to an arbitrary point $(x^3_0, u_0)$.

Thus we have found that in the subset of $CSI$ $CCNV$ spacetimes where $m_{33,r}$ and $\Gamma_{3n3}$ both vanish, there are no null Killing vectors other than $\ell$.  However, these spacetimes always admit the space-like Killing vector

\beq Y = n + m^3. \eeq

\end{subsection}

\end{section}

\begin{section}{Conclusions}

In this paper we have discussed the possibility of spacetimes with a covariantly constant null vector ($CCNV$) $\ell$  admitting an additional Killing vector distinct from $\ell$. Of particular interest are the set of $CNNV$ spacetimes with constant scalar curvature invariants ($CSI$ $CCNV$ spacetimes) containing another Killing vector $X$. Taking the general form for a $CCNV$ metric \eqref{CCNVKundt} we calculated the connection coefficients and Riemman tensor components. In doing so we found criteria for a $CCNV$ spacetime to be $CSI$; the transverse metric of \eqref{CCNVKundt} must be locally homogeneous.

With the connection coefficients calculated, we analyzed the Killing equations for a $CCNV$ metric. By performing a coordinate transform to eliminate $\hat W_3$ and another to rotate the frame so that $m_3$ is aligned with the spatial component of $X$, the equations simplify sufficiently in order to have expressions or equations for the metric functions ($H$, $\hat W_n$, $m_{ie}$) in terms of the Killing vector components, $X_a$. Thus we found all of the $CCNV$ spacetimes which admit a Killing vector $X$ and determined the three different forms $X$ may take; depending on whether $X_1$ is constant, linear in $u$ or a function of $u$ and $x^3$. These results are summarized in the two tables in section \ref{GeneralCCNV}. We examined the Killing Lie Algebra arising from $\{ \ell, X \} $.  The first two cases have closed Lie Algebras (in fact, they commute).  However, the last case may or may not be closed depending on the choice of $X_1(u, x^3)$.

By requiring that the transverse metric $g_{ef}$ is locally homogeneous, it is possible to make a coordinate transform so that it is independent of u (i.e. $g_{ef,u} = 0$), this causes the $B_{ij} = m_{ie,u}m_j^{~e}$ to vanish in the equations for the metric functions in section \ref{GeneralCCNV}, and so we find additional constraints on the $\hat W_n$, $H$. However, the equations involving $m_{ie}$ and $X_a$ must be reconsidered and the choice of transverse space may affect the choice of the Killing vector components $X_a$. In particular, if $X$ has a non-zero spatial component, $X_3$ and Proposition \eqref{prop:metricfn1} in subsection \ref{GeneralCCNVc2} does not hold, the transverse space must satisfy the following constraints

\beq m_{nr,3} &=& 0 \nonumber \\
 D_3(\Gamma_{3n3}) &=& 0 \nonumber \\
\Gamma^m_{~[3n]} \Gamma_{3m3} &=& 0 \nonumber \\
\Gamma^p_{~[nm]}\Gamma_{3p3} &=&  D_{n}\Gamma_{3m3} - D_{m}\Gamma_{3n3}. \nonumber \eeq

\noindent But if $X_3 = 0$, we may choose $g_{ef}$ to be any locally homogeneous space or if Proposition \eqref{prop:metricfn1} holds the transverse space must be locally homogenous and of the form \eqref{mei}. The tables at the end of section \ref{CSICCNV} summarize our results.

As an application of sectons \ref{GeneralCCNV} and \ref{CSICCNV} we studied the set of $CCNV$ spacetimes admitting a special Killing vector, $X$, that is non-spacelike for all values of the coordinate $v$. This requirement forces $X$ to be of type A; that is, $X_1$ is constant. As an illustration we presented a very simple example, a $CSI$ $CCNV$ spacetime admitting a null Killing vector where the transverse space satisfies $m_{3r,3} = 0$ and $\Gamma_{3n3} = 0$, in which the additional null Killing vector will be $n$ (up to the addition of a spacelike Killing vector).
%
%
%
%

\end{section}

\end{document}